\documentclass[10pt, onecolumn]{IEEEtran}

\date{}
\usepackage{amsfonts}
\usepackage{amsmath}
\usepackage{graphicx}
\usepackage{hyperref}
\usepackage{url}

\usepackage[a4paper]{geometry}
\usepackage{graphicx}
\usepackage{amsmath}
\usepackage{amssymb}
\usepackage{amsfonts}
\usepackage{amsopn,ntheorem}
\usepackage{tikz}

\usetikzlibrary{fit}					% fitting shapes to coordinates
\usetikzlibrary{backgrounds}	
\newcommand{\tr}{\text{tr}}

\newtheorem{theorem}{Theorem}%[section]
\newtheorem{lemma}{Lemma}

\newtheorem{corollary}{Corollary}

\newenvironment{proof}{{\noindent{\bf Proof:}}}{$\hfill\Box$}

\newcommand{\ket}[1]{|#1\rangle}
\newcommand{\bra}[1]{\langle#1|}

\def\det{\text{\rm{det}}}
\def\tr{\text{\rm{tr}}}
\def\ent{\text{\rm{ent}}}

\ifCLASSINFOpdf
\else
\fi

\begin{document}

\title{\huge Maximal Entanglement -- A New Measure of Entanglement}

\author{\IEEEauthorblockN{\large Salman Beigi}\\

\IEEEauthorblockA{\normalsize School of Mathematics,
Institute for Research in Fundamental Sciences (IPM),
Tehran, Iran}
}


\maketitle

\begin{abstract}
Maximal correlation is a measure of correlation for bipartite distributions. This measure has two intriguing features: (1) it is monotone under local stochastic maps; (2) it gives the same number when computed on i.i.d.\ copies of a pair of random variables. This measure of correlation has recently been generalized for bipartite quantum states, for which the same properties have been proved. In this paper, based on maximal correlation, we define a new measure of entanglement which we call maximal entanglement. We show that this measure of entanglement is faithful (is zero on separable states and positive on entangled states), is monotone under local quantum operations, and gives the same number when computed on tensor powers of a bipartite state. 

\end{abstract}

\IEEEpeerreviewmaketitle

%%%%%%%%%%%%%%%%%%%%%%%%%%%%%%%%%%%%%%%%%%%%%%%

\section{Maximal correlation}
Suppose that two parties, Alice and Bob, respectively receive $a\in \mathcal A$ and $b\in \mathcal B$ sampled from a joint distribution $p_{AB}$. Their goal is to respectively output $c\in \mathcal C$ and $d\in \mathcal D$ with some predetermined distribution $q_{CD}$. The question is whether this goal is achievable, assuming that there is no communication link between Alice and Bob. 
We call this problem the local transformation problem (transformation of a bipartite distribution to another one under local stochastic maps).
Deciding whether local transformation is possible, is a hard problem in general due to the non-linearity of the problem, especially when the size of $\mathcal A\times \mathcal B$ is large. In this case one is willing to obtain necessary or sufficient conditions for the problem.

Local stochastic maps (channels) cannot transform a less-correlated bipartite distribution, to a more-correlated one. So we may find bounds on this problem using measures of correlation. For instance, due to the data processing inequality, if the mutual information between $A$ and $B$, (denoted by $I(A; B)$), is less than $I(C; D)$ then $p_{AB}$ cannot be transformed to $q_{CD}$ under local operations.  

Suppose that Alice and Bob receive multiple samples from the source distribution $p_{AB}$, and want to generate only one sample from $q_{CD}$. That is, their inputs are $a^n\in \mathcal A^n$ and $b^n\in \mathcal B^n$, for some arbitrary $n$, sampled with probability 
$p^n(a^n, b^n):=\prod_{i} p(a_i, b_i)$, and they want to output $c$ and $d$ with probability $q(c, d)$. Mutual information in this case does not give any bound simply because $I(A^n;B^n)=nI(A;B)$ is greater than $I(C;D)$ for sufficiently large $n$ assuming that $I(A;B)$ is non-zero. In general, additive (or even ``weakly additive") measures of correlation do not give any bound for the latter problem.  

There is a measure of correlation called \emph{maximal correlation} which is first introduced by Hirschfeld~\cite{Hirschfeld} and Gebelein~\cite{Gebelein} and then studied by R\'enyi~\cite{Renyi1, Renyi2}. Later it is studied in~\cite{Witsenhausen} by Witsenhausen, and in the recent works~\cite{Kumar1, Kumar2, KamathAnantharam, KangUlukus, Anantharametal, Polyanskiy}. 
Maximal correlation is defined as follows:
\begin{align}
\mu(p_{AB})=\max\,\, & \mathbb E[fg]\label{eq:max-classic}\\
& \mathbb E[f] = \mathbb E[g] =0,\nonumber\\
& \mathbb E[f^2] = \mathbb E[g^2] =1,\nonumber
\end{align}
in which $f$ and $g$ are functions on $\mathcal A$ and $\mathcal B$ respectively. $\mu(p_{AB})$ is indeed the maximum of the expectation value of the product of two random variables that have zero mean and variance $1$. 

Maximal correlation is not additive. Indeed this measure of correlation, has the intriguing property that $\mu(p_{A^nB^n}^n)=\mu(p_{AB})$. Moreover, as a measure of correlation it satisfies the data processing inequality. Putting these together we conclude that if $\mu(p_{AB})<\mu(q_{CD})$ then local transformation of $p_{AB}$ to $p_{CD}$ is impossible even when an arbitrary large number of samples of the source $p_{AB}$ is available. 

Let us consider an example. Let $|\mathcal A|=|\mathcal B|=2$ and define the distribution $p^{(\epsilon)}_{AB}$ by 
\begin{align*}
p^{(\epsilon)}_{ab} = \begin{cases}
 \frac{1-\epsilon} 2&\quad a=b,\\
\frac \epsilon 2&\quad a\neq b.
\end{cases}
\end{align*}
Then for $0\leq \epsilon\leq 1/2$ we have $\mu\big(p_{AB}^{(\epsilon)}\big)=1-2\epsilon$. As a result, having even infinitely many samples from $p_{AB}^{(\epsilon)}$ we cannot locally generate samples from $p_{AB}^{(\delta)}$ if $0\leq \delta< \epsilon\leq 1/2$.

Maximal correlation has recently been defined for bipartite quantum states~\cite{Beigi12}. In this paper we first review the definition and properties of maximal correlation in the quantum case. Then based on maximal correlation, we define a new measure of entanglement, which we call \emph{maximal entanglement}. We show that this measure of entanglement is zero on separable states and positive on entangled ones, is monotone under local quantum operations, and gives the same number when computed on tensor powers of a bipartite state. This measure of entanglement, however, is not monotone under classical communication.

\section{Maximal correlation for quantum states}\label{sec:new}

Let us first fix some notations. The Hilbert space corresponding to a register $A$ is denoted by $\mathcal H_A$, and the space of linear operators acting on $\mathcal{H}_A$ by $\mathbf{L}(\mathcal{H}_A)$. In this paper we use Dirac's ket-bra notation for vectors: vectors of $\mathcal H_A$ are denoted by $\ket v$, and vectors in the dual space are denoted by $\bra w$; Then $\ket v\bra w$ is a linear operator and belongs to $\mathbf L(\mathcal H_A)$. The dimension of $\mathcal H_A$ is denoted by $d_A$ (in this paper we only consider finite dimensional quantum registers). The same notations are adopted for other quantum registers. 

We equip $\mathbf L(\mathcal H_A)$ with the Hilbert-Schmidt inner product $\langle X, Y\rangle:= \tr(X^\dagger Y)$ for $X, Y\in \mathbf L(\mathcal H_A)$, where $X^{\dagger}$ is the adjoint of $X$. %Then for $X\in \mathbf L(\mathcal H_A)$ its 2-norm is defined by $\|X\|_2 =\tr(X^{\dagger}X)^{1/2}$.

Now to define maximal correlation for a bipartite density matrix $\rho_{AB}\in \mathbf L(\mathcal H_A)\otimes \mathbf L(\mathcal H_B)$, we should think of $f, g$ in \eqref{eq:max-classic} as quantum observables, whose expectations are computed with Born's rule. Based on this intuition, maximal correlation of $\rho_{AB}$ is defined by
\begin{align}
\mu(\rho_{AB})=\max\,\, & \vert \tr(\rho_{AB} X_A\otimes Y^{\dagger}_{B})\vert\label{eq:con0}\\
& \tr(\rho_AX_A) = \tr(\rho_B Y_B) =0,\label{eq:con1}\\
& \tr(\rho_A X_AX_A^{\dagger}) = \tr(\rho_B Y_BY_B^{\dagger}) =1.\label{eq:con2}
\end{align}
Here $\rho_A=\tr_B(\rho_{AB})$ and $\rho_B=\tr_A(\rho_{AB})$ are the reduced density matrices on subsystems $A$ and $B$ respectively, and $X_A\in \mathbf{L}(\mathcal{H}_A)$ and $Y_B\in \mathbf{L}(\mathcal{H}_B)$. As shown in~\cite{Beigi12} in the above optimization one may assume that $X_A$ and $Y_B$ are hermitian.

To study properties of $\mu(\rho_{AB})$ it would be useful to define
$$\widetilde \rho_{AB} = (I_A\otimes \rho_B^{-1/2}) \rho_{AB} (\rho_A^{-1/2}\otimes I_B).$$
Here the inverses of $\rho_A$ and $\rho_B$ are defined on their supports (Moore-Penrose pseudo-inverse). In fact without loss of generality, by restricting $\mathcal H_A, \mathcal H_B$ to the supports of $\rho_A, \rho_B$ respectively, we may assume that the local density matrices are full-rank.

The following theorem proved in~\cite{Beigi12} is the quantum version of the results of~\cite{Kumar2} and~\cite{KangUlukus} connecting maximal correlation to certain Schmidt coefficients (see also \cite{Polyanskiy}).

\begin{theorem} \cite{Beigi12} \label{thm:schmidt}
The first Schmidt coefficient of $\widetilde \rho_{AB}$ as a vector in the bipartite Hilbert space $\mathbf{L}(\mathcal{H}_A)\otimes \mathbf{L}(\mathcal{H}_B)$ is equal to $1$, and its second Schmidt coefficient is equal to
$\mu(\rho_{AB})$.
\end{theorem}

Note that $\mathbf{L}(\mathcal{H}_A)$ and $\mathbf{L}(\mathcal{H}_B)$ are Hilbert spaces equipped with the Hilbert-Schmidt inner product. So Schmidt decomposition, and then Schmidt coefficients, of vectors in $\mathbf{L}(\mathcal{H}_A)\otimes \mathbf{L}(\mathcal{H}_B)$ are computed with respect to Hilbert-Schmidt inner product. To be more precise, the Schmidt decomposition of $\widetilde \rho_{AB}$ is of the form 
$$\widetilde \rho_{AB}=\sum_i \lambda_i M_i\otimes N_i,$$
where $\lambda_1\geq \lambda_2\geq \cdots\geq 0$ are Schmidt coefficients, and $M_i\in \mathbf L(\mathcal H_A)$ and $N_i\in \mathbf L(\mathcal H_B)$ are orthonormal bases, i.e., $\tr(M_i^{\dagger}M_j) = \tr(N_i^{\dagger}N_j)=\delta_{ij}$ where $\delta_{ij}$ is the Kronecker delta function. The above theorem states that 
$$\lambda_1=1 \text{ and  } \lambda_2=\mu(\rho_{AB}).$$

Let us examine this theorem in the classical case. Let $p_{AB}$ be a joint distribution. Let $P_{AB}$ be a $|\mathcal A|\times |\mathcal B|$-matrix whose $(a, b)$-th entry is equal to $p_{ab}$. Also let $P_A$ be the diagonal matrix with diagonal entries $p_{a}$. Define $P_B$ similarly. Define $\widetilde P_{AB} = P_A^{-1/2}P_{AB}P_B^{-1/2}$ which is the classical analogue of $\widetilde\rho_{AB}$ defined above. 
Then the first \emph{singular value} of $\widetilde P_{AB}$ is equal $1$ and its second singular value is equal to $\mu(P_{AB})$. For example, if $\mathcal A= \mathcal B=\{0,1\}$, then 
\begin{align*}
\widetilde P_{AB}=
\begin{pmatrix}
\frac{p_{00}}{\sqrt{(p_{00}+p_{01})(p_{00}+p_{10})}} & \frac{p_{01}}{\sqrt{(p_{00}+p_{01})(p_{01}+p_{11})}}\\
\frac{p_{10}}{\sqrt{(p_{10}+p_{11})(p_{00}+p_{10})}} & \frac{p_{11}}{\sqrt{(p_{10}+p_{11})(p_{01}+p_{11})}}
\end{pmatrix}.
\end{align*}
$\widetilde P_{AB}$ has two singular values one of which is $1$. So the other one is equal to 
\begin{align}\label{eq:max-det}
\mu(P_{AB}) = |\det \widetilde P_{AB}|.
\end{align}

Relating maximal correlation of $\rho_{AB}$ to Schmidt coefficients of $\widetilde \rho_{AB}$, we can now state the main properties of maximal correlation.

\begin{theorem} \cite{Beigi12}\label{thm:ab} $\mu(\cdot)$ satisfies the following properties:
\begin{enumerate}
\item[{\rm(a)}] $\mu(\rho_{AB}\otimes \sigma_{A'B'}) = \max\{ \mu(\rho_{AB}), \mu(\sigma_{A'B'})  \}$.
\item[{\rm(b)}] Let $\Phi_{A\rightarrow A'}, \Psi_{B\rightarrow B'}$ be completely positive trace-preserving super-operators. Let $\sigma_{A'B'} =\Phi\otimes \Psi(\rho_{AB})$. Then
$\mu(  \sigma_{A'B'}  )   \leq \mu(\rho_{AB})$.
\end{enumerate}
\end{theorem}

\begin{proof}
(a) This is a simple consequence of the fact that Schmidt coefficients of the tensor product of two vectors are equal to the pairwise products of Schmidt coefficients of the two vectors.

\noindent (b) Completely positive trace-preserving maps are compositions of isometries and partial traces. It is not hard to see that local isometries do not change $\mu(\rho_{AB})$. So we only need to show that $\mu(\sigma_{AB}) \leq \mu(\sigma_{AA'BB'})$. This inequality holds because in the definition of maximal correlation for $\sigma_{AA'BB'}$ we may restrict $X_{AA'}, Y_{BB'}$ to have the form $X_{AA'}=X_A\otimes I_{A'}$ and $Y_{BB'}=Y_B\otimes I_{B'}$. With these restrictions we obtain $\mu(\sigma_{AB})$ as the optimal value.

\end{proof}

The following corollary is a direct consequence of the above theorem.

\begin{corollary} \label{cor:1} Suppose that $\rho_{AB}^{\otimes n}$, for some $n$, can be locally transformed to $\sigma_{EF}$. Then
$$\mu(\rho_{AB}) \geq \mu(\sigma_{EF}).$$
\end{corollary}

Let us workout an example. Let $\ket \psi_{AB} = \frac{1}{\sqrt{2}}(\ket{00} + \ket {11})$ be the Bell state on two qubits. Define 
\begin{align}\label{eq:epsilon-bell}
\rho^{(\epsilon)}_{AB} =(1-\epsilon)\,\ket \psi\bra\psi_{AB}+ \epsilon \frac{I_{AB}}{4},
\end{align}
where $0\leq \epsilon\leq 1$ and $I_{AB}/4$ is the maximally mixed state. It is not hard to see that $\mu\big(\rho^{(\epsilon)}_{AB}\big) = 1- \epsilon$ (see~\cite{Beigi12} for details). 
Now using the above corollary, having an arbitrary large number of copies of $\rho_{AB}^{(\epsilon)}$ one cannot transform them into one copy of $\rho_{AB}^{(\delta)}$ under local transformations if $\epsilon> \delta$.

The following theorem characterizes the extreme values of maximal correlation. 

\begin{theorem} \cite{Beigi12} \label{thm:extreme}
$0\leq \mu(\rho_{AB})\leq 1$ and the followings hold:
\begin{itemize} 
\item[\rm{(a)}] $\mu(\rho_{AB})=0$ if and only if $\rho_{AB}= \rho_A\otimes \rho_B$.
\item[\rm{(b)}] $\mu(\rho_{AB})=1$ if and only if there exist local measurements $\{M_0, M_1\}$ on $A$, and $\{N_0, N_1\}$ on $B$ such that $0<\tr\left(\rho_{AB} M_0\otimes N_0\right)<1$, and $$\tr\left(\rho_{AB} (M_0\otimes N_1)\right) = \tr\left( \rho_{AB} (M_1 \otimes N_0)  \right) =0.$$
\item[\rm{(c)}] For a pure state $\rho_{AB}=\ket \psi\bra \psi_{AB}$, we have $\mu(\rho_{AB})=0$ if $\ket \psi_{AB}$ is a product state, and $\mu(\rho_{AB})=1$ if $\ket \psi_{AB}$ is entangled.
\end{itemize}
\end{theorem}

Part (b) of this theorem is the quantum version of Witsenhausen's result~\cite{Witsenhausen} that maximal correlation of a bipartite distribution is equal to 1 iff the distribution has a ``common data." 

We finish this section by proving two lemmas which will be used in the next section.

\begin{lemma}\label{lem:large}
Let $\rho_{AB}$ be a bipartite density matrix such that $\bra \psi \rho_{AB} \ket \psi \geq1-\epsilon$, where $\ket{\psi}_{AB}=\frac{1}{\sqrt 2}(\ket{00}+\ket{11})$. Then 
$$\mu(\rho_{AB})\geq 1-2\epsilon.$$
\end{lemma}

\begin{proof} There is nothing to prove for $\epsilon\geq 1/2$ since maximal correlation is non-negative. So we assume that $0\leq  \epsilon<1/2$.
Let us define 
\begin{align*}
\ket{\psi_i} = \frac{1}{\sqrt 2}\big(\ket{00}+ (-1)^i \ket{11}\big), \qquad i=0, 1
\\
\ket{\psi_i} = \frac{1}{\sqrt 2}\big(\ket{01}+ (-1)^i \ket{10}\big)\qquad i=2,3.
\end{align*}
Then $\ket{\psi_0}=\ket \psi$  and $\{\ket{\psi_{i}}: 0\leq i\leq 3 \}$ is an orthonormal basis for the space of two qubits. Therefore there are $c_{ij}\in \mathbb C$ such that 
$$\rho_{AB}=\sum_{i,j=0}^3 c_{ij}\ket{\psi_i}\bra{\psi_j}.$$
By assumption we have $c_{00}\geq1-\epsilon$. Moreover, since $\rho_{AB}$ is a density matrix $\sum_{i} c_{ii}=1$ and 
\begin{align}\label{eq:ineq20}
|c_{01}|^2\leq c_{00}c_{11}, \qquad |c_{12}|^2\leq c_{22}c_{33}.
\end{align}

Now suppose that we measure each qubit $A, B$ in the computational basis $\{\ket 0, \ket 1\}$. We let $p_{EF}$ be the outcome distribution. We have 
\begin{align*}
p_{00} & = \frac{1}{2}(c_{00}+ c_{11} + c_{01} + c_{10})\\
p_{11} & = \frac{1}{2}(c_{00}+ c_{11} - c_{01} - c_{10})\\
p_{01} & = \frac{1}{2}(c_{22}+ c_{33} + c_{23} + c_{32})\\
p_{10} & = \frac{1}{2}(c_{22}+ c_{33} - c_{23} - c_{32}).
\end{align*}
By the monotonicity of maximal correlation under local measurements (Theorem~\ref{thm:ab}) we have $\mu(\rho_{AB})\geq \mu(p_{EF})$. Moreover, maximal correlation of $p_{EF}$ can be computed using~\eqref{eq:max-det}. Thus we obtain
\begin{align*}
\mu(\rho_{AB})&\geq \frac{|p_{00}p_{11} - p_{01}p_{10}|}{\big[{(p_{00}+p_{01})(p_{00}+p_{10})(p_{10}+p_{11})(p_{01}+p_{11})}\big]^{\frac 1 2}}.
\end{align*}

Let $q=c_{00}+c_{11}$ (which gives $c_{22}+c_{33}=1-q$), $a=c_{01}+c_{10}$ and $b=c_{23}+c_{32}$. Then $q\geq 1-\epsilon$, and by~\eqref{eq:ineq20} we have $|a|\leq q$ and $|b|\leq 1-q$. Indeed there is a stronger upper bound on $|a|$; Using $c_{00}\geq 1-\epsilon$ and $c_{00}+c_{11}=q$ we have
\begin{align*}
|a|  \leq 2|c_{01}| \leq 2\sqrt{c_{00}c_{11}} \leq 2\sqrt{(1-\epsilon)(q-1+\epsilon)}.
\end{align*}
Using this inequality, and the fact that $q\geq 1-\epsilon \geq 1/2$ we find that 
$|a|^2\leq 2q-1.$

Now rewriting the above bound in terms of these new variables we obtain
\begin{align*}
\mu(\rho_{AB})&\geq \frac{|q^2-a^2 - (1-q)^2 +b^2|}{\big[{(1+ a+b)(1+ a-b)(1- a+b)(1- a-b)}\big]^{1/2}}\\
&= \frac{|2q-1 +b^2-a^2|}{\big[ 1+ a^4 + b^4- 2a^2b^2 - 2a^2 - 2b^2   \big]^{1/2}}\\
&= \frac{2q-1 +b^2-a^2}{\big[ 1+ a^4 + b^4- 2a^2b^2 - 2a^2 - 2b^2   \big]^{1/2}},
\end{align*}
where in the last line we use $|a|^2\leq 2q-1$. The numerator of the latter bound is obviously an increasing function of $|b|$. Moreover using the fact that $|b|^2\leq 1-q\leq 1+|a|^2$, it is not hard to see that the denominator is a decreasing function of $|b|$. Therefore 
\begin{align*}
\mu(\rho_{AB})&\geq \frac{2q-1-a^2}{1-a^2}\\
&= 1+ \frac{2(q-1)}{1-a^2}\\
&\geq 1+2(q-1)\\
&\geq 1-2\epsilon.
\end{align*}

\end{proof}

\begin{lemma}\label{lem:mu-cont}
Let $\{\sigma_{AB}^{(n)}\}$ be a sequence of density matrices and 
suppose that $\lim_{n\rightarrow \infty} \sigma_{AB}^{(n)} = \rho_{AB}$ and that $\lim_{n\rightarrow \infty} \mu(\sigma_{AB}^{(n)}) = \lambda$. Then we have $\mu(\rho_{AB})\leq \lambda$. In particular if $\lambda=0$, then $\rho_{AB}$ is a product state. 
\end{lemma}

\begin{proof} Let $X, Y$ be the operators that achieve the optimal value $\mu(\rho_{AB})$ in \eqref{eq:con0}-\eqref{eq:con2}. For every $n$ define 
\begin{align*}
X_n:=\frac{X- \tr(\sigma_A^{(n)})}{\big[  \tr\big(\sigma_A^{(n)} XX^{\dagger}\big) - |\tr(\sigma_A^{(n)} X)|^2   \big]^{1/2}},
\end{align*}
and 
\begin{align*}
Y_n:=\frac{Y- \tr(\sigma_B^{(n)})}{\big[  \tr(\sigma_B^{(n)} YY^{\dagger}) - |\tr(\sigma_B^{(n)} Y)|^2   \big]^{1/2}}.
\end{align*}
Then for every $n$ we have $\tr(\sigma_A^{(n)}X_n) = \tr(\sigma_B^{(n)}Y_n)=0$ and $\tr(\sigma_A^{(n)}X_nX_n^{\dagger})= \tr(\sigma_B^{(n)} YY^{\dagger})=1$. On the other hand, it is easy to see that 
$$\lim_{n\rightarrow \infty} |\tr \big(\sigma_{AB}^{(n)} X_n\otimes Y_n^{\dagger}\big)  | = |\tr\big(\rho_{AB}X\otimes Y^{\dagger}\big)|=\mu(\rho_{AB}).$$
This means that if $\mu(\rho_{AB})\geq \lambda+\epsilon$ for some $\epsilon>0$, then for sufficiently large $n$, $\mu\big(\sigma_{AB}^{(n)}\big) \geq \lambda+\epsilon/2$, which is a contradiction since $\lim_{n\rightarrow \infty} \mu(\sigma_{AB}^{(n)}) = \lambda$. Therefore $\mu(\rho_{AB})\leq \lambda$.

\end{proof}

Note that, by the assumption of the above lemma we cannot conclude the equality of $\mu(\rho_{AB})$ and $\lambda$. For example, define the distributions $p_{AB}^{(n)}$ by $p_{00}^{(n)} = 1-1/n$, $p_{11}^{(n)}=1/n$ and $p_{01}^{(n)}=p_{10}^{(n)}=0$. Then for every $n$ we have $\mu(p_{AB}^{(n)})=1$, but this sequence of distributions converges to $q_{AB}$ with $q_{00}=1$ and $q_{01}=q_{10}=q_{11}=0$, for which we have $\mu(q_{AB})=0$.

%%%%%%%%%%%%%%%%%%%%%%%%%%%%%%%%%%%

\section{Maximal entanglement}

In this section we define a measure of entanglement in terms of maximal correlation. By part (c) of Theorem~\ref{thm:extreme}, maximal correlation is already a measure of entanglement on pure states; This is the measure that is $0$ for pure product states and $1$ for pure entangled states. Nevertheless, maximal correlation is not a measure of entanglement since there are bipartite (classical) distributions whose maximal correlation is $1$. Convex roof extension is the idea that is usually applied to construct measures of entanglement in such situations. In our case however, convex roof does not result in a measure that satisfies our desired properties. So we propose the following definition:
$$\mu_{\ent}(\rho_{AB}):= \inf \max_i \mu(\tau_{AB}^{(i)}),$$
where the infimum is taken over all decompositions 
$$\rho_{AB}=\sum_i p_i \tau_{AB}^{(i)}$$ 
where $p_i\geq 0$ and $\tau_{AB}^{(i)}$'s are density matrices. We call $\mu_\ent(\cdot)$ maximal entanglement.
From the definition we clearly have
\begin{enumerate}
\item[\rm{(i)}] $0\leq \mu_{\ent}(\rho_{AB})\leq 1$.
\item[\rm{(ii)}] $\mu_{\ent}(\rho_{AB})\leq \mu(\rho_{AB})$.
\item[\rm{(iii)}] $\mu_{\ent}(\rho_{AB})=0$ for all separable states $\rho_{AB}$.
\item[\rm{(iv)}] If $\rho_{AB}$ is pure then $\mu_\ent(\rho_{AB}) = \mu(\rho_{AB})$.
\item[\rm{(v)}] Maximal entanglement is quasi-convex, i.e., we have $\mu_\ent\big(\sum_i\lambda_i \rho_{AB}^{(i)}\big)\leq \max_i \mu_\ent(\rho_{AB}^{(i)})$.
\end{enumerate}
Statement (iii) holds simply because every separable state can be written in the form $\rho_{AB} = \sum_i p_i \tau_A^{(i)}\otimes \tau_B^{(i)}$, and for all $i$ we have $\mu(\tau_A^{(i)}\otimes \tau_B^{(i)})=0$. Statement (iv) holds because pure states essentially have a unique decomposition. 

The following theorems state the main properties of maximal entanglement. 

\begin{theorem}\label{thm:max-ent-data} 
Suppose that $\sigma_{A'B'}=\Phi\otimes \Psi(\rho_{AB})$ where $\Phi_{A\rightarrow A'}$ and $\Psi_{B\rightarrow B'}$ are completely positive trace-preserving maps. Then 
$$\mu_{\ent}(\rho_{AB})\geq \mu_{\ent}(\sigma_{A'B'}).$$
\end{theorem}

\begin{proof} Let $\rho_{AB}=\sum_i p_i \tau_{AB}^{(i)}$ be a decomposition of $\rho_{AB}$. Then 
$\sigma_{A'B'}=\sum_i p_i \Phi\otimes \Psi(\tau_{AB}^{(i)})$ is a decomposition of $\sigma_{A'B'}$. Using Theorem~\ref{thm:ab} for every $i$ we have
\begin{align*}
\mu\big(\tau_{AB}^{(i)}\big) \geq  \mu\big(\Phi\otimes \Psi(\tau_{AB}^{(i)})\big).
\end{align*}
Then taking maximum over $i$, and then infimum over all decompositions of $\rho_{AB}$ we obtain the result. 
\end{proof}

\begin{theorem}\label{thm:max-ent-tensorize} For all density matrices $\rho_{AB}, \sigma_{A'B'}$ we have
$\mu_{\ent}(\rho_{AB}\otimes \sigma_{A'B'}) = \max\{\mu_{\ent}(\rho_{AB}), \mu_{\ent}(\sigma_{A'B'})\}$.
\end{theorem}

\begin{proof}  Using the data processing inequality (Theorem~\ref{thm:max-ent-data}) by considering \emph{local} partial trace maps $\tr_{AB}(\cdot)$ and $\tr_{A'B'}(\cdot)$ we obtain 
$$\mu_{\ent}(\rho_{AB}\otimes \sigma_{A'B'}) \geq \max\{\mu_{\ent}(\rho_{AB}), \mu_{\ent}(\sigma_{A'B'})\}.$$
For the other direction consider decompositions $\rho_{AB}=\sum_i p_i \tau_{AB}^{(i)}$  and $\sigma_{A'B'} = \sum_j q_j \xi_{A'B'}^{(j)}$. Then $$\rho_{AB}\otimes \sigma_{A'B'}= \sum_{i,j} p_iq_j \tau_{AB}^{(i)}\otimes \xi_{A'B'}^{(j)},$$
is a decomposition of $\rho_{AB}\otimes \sigma_{A'B'}$. As a result 
\begin{align*}
\mu_\ent(\rho_{AB}\otimes \sigma_{A'B'}) &\leq \max_{i,j} \mu\big(\tau_{AB}^{(i)}\otimes \xi_{A'B'}^{(j)}\big)\\
& = \max_{i,j} \max\left\{ \mu\big(\tau_{AB}^{(i)}\big),  \mu\big(\xi_{A'B'}^{(j)}\big)     \right\}.
\end{align*}
Taking infimum over all decompositions of $\rho_{AB}$ and $\sigma_{A'B'}$ gives the desired result. 

\end{proof}

The above two theorems imply the following analogue of Corollary~\ref{cor:1} for maximal entanglement.

\begin{corollary}\label{cor:2}
Suppose that $\rho_{AB}^{\otimes n}$, for some $n$, can be locally transformed to $\sigma_{EF}$. Then
$$\mu_\ent(\rho_{AB}) \geq \mu_\ent(\sigma_{EF}).$$
\end{corollary}

\begin{theorem}\label{thm:max-ent-faithful} Maximal entanglement is faithful, i.e., $\mu_{\ent}(\rho_{AB})=0$ if and only if $\rho_{AB}$ is separable. 
\end{theorem}

\begin{proof} We already know that separability of $\rho_{AB}$ implies $\mu_{\ent}(\rho_{AB})=0$. So we need to show that if $\mu_{\ent}(\rho_{AB})=0$ then $\rho_{AB}$ is separable.  

By definition $\mu_{\ent}(\rho_{AB})=0$ means that for every $\epsilon>0$ there is a decomposition $\rho_{AB}=\sum_i p_i \tau_{AB}^{(i)}$ such that $\mu(\tau_{AB}^{(i)})\leq \epsilon$ for all $i$. Let $\mathcal S_{\epsilon}$ be the set of all density matrices whose maximal correlation is at most $\epsilon$:
$$\mathcal S_{\epsilon}:=\{\sigma_{AB}:\, \mu(\sigma_{AB})\leq \epsilon\}.$$
 Therefore, $\mu_{\ent}(\rho_{AB})=0$ implies that $\rho_{AB}$ is in the convex hull of $\mathcal S_{\epsilon}$ for every $\epsilon>0$. Now using Carath\'eodory's theorem, for every $\epsilon>0$ there is a decomposition 
\begin{align*}
\rho_{AB} = \sum_{i=1}^{M} p_i^{(\epsilon)} \tau_{AB}^{(i, \epsilon)},
\end{align*}
such that $\mu\big(\tau_{AB}^{(i, \epsilon)}\big)\leq \epsilon$ for all $i$, and $M=M(d_A, d_B)$ is a constant independent of $\epsilon$. Now by a standard compactness argument there is a sequence  $\epsilon_1> \epsilon_2> \cdots$ with $\lim_{n\rightarrow \infty} \epsilon_n=0$ such that 
$$\Big(p_1^{(\epsilon_n)}, \dots, p_M^{(\epsilon_n)}, \tau_{AB}^{(1, \epsilon_n)}, \dots, \tau_{AB}^{(M, \epsilon_n)}\Big),$$
converges to some
$$\Big(p_1^{(0)}, \dots, p_M^{(0)}, \tau_{AB}^{(1, 0)}, \dots, \tau_{AB}^{(M, 0)}\Big).$$
This tuple then gives a valid decomposition 
$$\rho_{AB}=\sum_{i=1}^M p_i^{(0)} \tau_{AB}^{(i, 0)},$$
where $\{p_i^{(0)}: 1\leq i\leq M\}$ is a probability distribution and each $\tau_{AB}^{(i, 0)}$ is a density matrix. On the other hand, since $\lim_{n\rightarrow \infty} \tau_{AB}^{(i, \epsilon_n)} = \tau_{AB}^{(i, 0)}$ and that $\mu(\tau_{AB}^{(i, \epsilon_n)})\leq \epsilon_n$, using  Lemma~\ref{lem:mu-cont} we find that $\mu(\tau_{AB}^{(i, 0)})=0$. Then by part (a) of Theorem~\ref{thm:extreme}, $\tau_{AB}^{(i, 0)}$ is a product state for every $i$. As a result $\rho_{AB}$ is separable.  

\end{proof}

We showed in Theorem~\ref{thm:max-ent-data} that maximal entanglement is monotone under local (quantum) operations. 
As a measure of entanglement one may expect that maximal entanglement is also monotone under classical communication too. However, this is not the case. 

\begin{theorem}\label{thm:max-ent-locc} Maximal entanglement is not monotone under classical communication. 
\end{theorem}

\begin{proof} We know that any two-qubit entangled state is distillable~\cite{HHH}. That is, having copies of an entangled two-qubit state $\rho_{AB}$, approximations of the maximally entangled state $\ket\psi_{AB} = \frac{1}{\sqrt 2}(\ket{00}+\ket{11})$ can be distilled using local quantum operations and classical communication (LOCC). However, we know that there are two-qubit entangled states $\rho_{AB}$ such that $\mu_{\ent}(\rho_{AB})=\mu_{\ent}(\rho_{AB}^{\otimes n})\leq \mu(\rho_{AB})<1=\mu_\ent(\ket\psi\bra \psi_{AB})$ (see for instance the example after Corollary~\ref{cor:1}). On the other hand, local quantum operations do not increase maximal entanglement. Therefore, it should be classical communication that (sometimes) increases maximal entanglement. 

\end{proof}

Computing maximal entanglement $\mu_\ent(\rho_{AB})$ seems a hard problem in general since we need to take an infimum over all decompositions of $\rho_{AB}$. Finding upper bounds on $\mu_{\ent}(\rho_{AB})$ however, is easy; We only need to pick a decomposition to find an upper bound. In particular we have $\mu_\ent(\rho_{AB})\leq \mu(\rho_{AB})$. In the following we present some ideas that may serve as a useful tool
for proving lower bounds on maximal entanglement.

In the previous section we observed that the maximal correlation of $\rho_{AB}^{(\epsilon)}$ defined by~\eqref{eq:epsilon-bell} is equal to $\mu\big(\rho_{AB}^{(\epsilon)}\big)=1-\epsilon$. Computing the maximal entanglement of these states however, does not seem easy. By the partial transpose test we know that $\rho_{AB}^{(\epsilon)}$ is separable if and only if $\epsilon\leq 2/3$ ~\cite{Peres, HHH2}. Then using the faithfulness of maximal entanglement (Theorem~\ref{thm:max-ent-faithful}), $\mu_\ent(\rho_{AB}^{(\epsilon)})=0$ if $\epsilon\geq 2/3$ and $\mu_\ent(\rho_{AB}^{(\epsilon)})>0$ otherwise. 

In the following we present a characterization of $\mu_\ent\big(\rho_{AB}^{(\epsilon)}\big)$ in terms another optimization problem.

\begin{theorem}\label{thm:max-ent-isotropic} Define
$$\lambda(\epsilon)= \min\{\mu(\sigma_{AB}): \bra \psi \sigma_{AB}\ket\psi \geq 1-3\epsilon/4\},$$
where $\ket \psi_{AB}= \frac{1}{\sqrt 2}(\ket{00}+\ket{11})$. Then $\mu_{\ent}\big(\rho_{AB}^{(\epsilon)}\big)=\lambda(\epsilon)$, where $\rho_{AB}^{(\epsilon)} = (1-\epsilon) \ket \psi\bra \psi_{AB} + \epsilon I_{AB}/4$. Moreover, $\lambda(\epsilon)=0$ for $\epsilon\geq 2/3$, and  for $\epsilon<2/3$ we have
$$1-3\epsilon/2\leq \lambda(\epsilon)\leq 1-\epsilon.$$
\end{theorem}

\begin{proof}
For every decomposition $\rho_{AB}^{(\epsilon)}=\sum_i p_i \sigma_{AB}^{(i)}$, we have
\[\bra \psi \rho_{AB}^{(\epsilon)} \ket \psi = \sum_i p_i \bra \psi\sigma_{AB}^{(i)}\ket \psi = 1-3\epsilon/4.\]
Then there exists $i^*$ such that $\bra \psi \sigma_{AB}^{(i^*)}\ket \psi\geq 1-3\epsilon/4$. Therefore, by definition $\mu\big(\sigma_{AB}^{(i^*)}\big)\geq \lambda(\epsilon)$. This means that for any such decomposition we have $\max_i \mu(\sigma_{AB}^{(i)})\geq \lambda(\epsilon)$, and then $\mu_{\ent}(\rho_{AB}^{(\epsilon)})\geq \lambda(\epsilon)$.

For the other direction let $\sigma_{AB}$ be a state with $\bra \psi \sigma_{AB}\ket \psi \geq 1-3\epsilon/4$ and $\mu(\sigma_{AB})=\lambda(\epsilon)$. Let us define
\begin{align}\label{eq:iso-decomp}
\tau_{AB}:= \int (U\otimes U^*) \sigma_{AB} (U\otimes U^{*})^{\dagger} \text{d}U,
\end{align}
where $\text{d}U$ denotes the Haar measure on the unitary group. $\tau_{AB}$ is an \emph{isotropic} state~\cite{HHisotropic}, i.e., for every unitary $U$ we have $(U\otimes U^{*})\tau_{AB}(U\otimes U^{*})^{\dagger} = \tau_{AB}$. Any isotropic state is of the form
$$\tau_{AB} = (1-\delta)\ket \psi \bra \psi_{AB} + \delta I_{AB}/4 = \rho_{AB}^{(\delta)}.$$
On the other hand, since $U\otimes U^*\ket\psi_{AB}=\ket\psi_{AB}$ for every unitary $U$, we have 
$$1-3\delta/4=\bra \psi\tau_{AB}\ket \psi = \bra \psi \sigma_{AB}\ket \psi \geq 1-3\epsilon/4.$$
Therefore $\delta\leq \epsilon$.
%This means that $\tau_{AB} = \rho_{AB}^{(1-t)}$ which we know that $t\geq 1-\epsilon$.

Now observe that~\eqref{eq:iso-decomp} is already a decomposition of $\tau_{AB}=\rho_{AB}^{(\delta)}$. Moreover, maximal correlation does not change under local unitaries. Therefore, we have 
$$\mu_{\ent}\big(\rho_{AB}^{(\delta)}\big)\leq \mu(\sigma_{AB})=\lambda(\epsilon).$$
To finish the prove we only need to notice that $\delta\leq \epsilon$ and then $\mu_{\ent}(\rho_{AB}^{(\delta)})\geq \mu_\ent(\rho_{AB}^{(\epsilon)})$. The latter inequality holds because 
$$\rho_{AB}^{(\epsilon)} = \frac{1-\epsilon}{1-\delta}\rho_{AB}^{(\delta)} + \frac{1+\epsilon-\delta}{1-\delta}I/4,$$
and that maximal entanglement is quasi-convex. 

$\lambda(\epsilon)=0$ for $\epsilon\geq 2/3$ is implied by the partial transpose test as discussed above. 
The lower bound $\lambda(\epsilon)\geq 1-3\epsilon/2$ is an immediate consequence of Lemma~\ref{lem:large}. The upper bound $\lambda(\epsilon)\leq 1-\epsilon$ is because $\mu_\ent\big(\rho_{AB}^{(\epsilon)}\big)\leq \mu\big(\rho_{AB}^{(\epsilon)}\big)$.

\end{proof}

\section{Concluding remarks}

In this paper we defined a measure of entanglement $\mu_\ent(\cdot)$ for quantum states that is faithful, monotone under local quantum operations, and gives the same number when computed on tensor powers of a bipartite state. The latter property, in particular, implies that maximal entanglement is neither super-additive nor monogamous. 

As we saw in Theorem~\ref{thm:max-ent-locc} a measure of entanglement with the above properties is either not monotone under classical communication or achieves its maximum value on all distillable states.

%%%%%%%%%%%%%%%%%%%%%%%%%%%%%%%%%%%%%%%%%%%%%%%

\end{document}